\documentclass[envcountsame]{llncs}

\newcommand{\longversion}[1]{#1}
\newcommand{\shortversion}[1]{}

\usepackage{times}

\longversion{
\usepackage[totalwidth=450pt,totalheight=640pt]{geometry}
\setlength{\oddsidemargin}{5mm}
\setlength{\evensidemargin}{5mm}
}
 
\usepackage{tikz}
\usepackage{url}
\usepackage{amsmath,amssymb}
\usepackage{xspace,tikz}
\usepackage{mathrsfs,verbatim}
\usepackage{enumitem}

\usepackage{thmtools,thm-restate}

\spnewtheorem{fact}{Fact}{\bfseries}{\itshape}

\newcommand{\Nat}{\mathbb{N}}
\newcommand{\hy}{\hbox{-}\nobreak\hskip0pt}

\def\hy{\hbox{-}\nobreak\hskip0pt} 

\newcommand{\SB}{\{\,} \newcommand{\SM}{\;{:}\;} \newcommand{\SE}{\,\}}
 
\newcommand{\Card}[1]{|#1|}

\newcommand{\CCC}{\mathcal{C}} 
\newcommand{\WWW}{\mathcal{W}} 
\newcommand{\RRR}{\mathcal{R}}

\newcommand{\NP}{\text{\normalfont NP}}

\newcommand{\mtext}[1]{\text{\normalfont\itshape #1}} 

\newcommand{\gleq}{\mathop{\diamondsuit}}

\newcommand{\rwc}[1]{\mtext{rwc}_{#1}}
\newcommand{\vcn}{\mtext{vcn}}
\newcommand{\nd}{\mtext{nd}}
\newcommand{\tw}{\mtext{tw}}
\newcommand{\rw}{\mtext{rw}}

\def\mx#1{\mbox{\boldmath$#1$}}

\def\MS#1{\mbox{MSO}}

\newcommand{\MSO}{\mbox{MSO}\xspace}

\def\lMS{\mbox{\textsc{LinEMS}}}

\date{}

\let\doendproof\endproof
\renewcommand\endproof{~\hfill$\qed$\doendproof}

\begin{document}

\pagestyle{plain}

\title{Meta-Kernelization with Structural Parameters\longversion{\thanks{Research supported by the European Research Council (ERC),
    project COMPLEX REASON 239962.}}}

\author{Robert Ganian \and Friedrich Slivovsky \and Stefan Szeider}

\institute{
 Institute of Information Systems,
 Vienna University of Technology, Vienna, Austria
\email{rganian@gmail.com,fslivovsky@gmail.com,stefan@szeider.net}
}

\maketitle

\begin{abstract}
  \noindent Meta-kernelization theorems are general results that provide
  polynomial kernels for large classes of parameterized problems.  The known
  meta-kernelization theorems, in particular the results of Bodlaender et
  al.~(FOCS'09) and of Fomin et al.~(FOCS'10), apply to optimization problems
  parameterized by \emph{solution size}. We present meta-kernelization
  theorems that use a \emph{structural parameters} of the input and not the
  solution size.  Let $\CCC$ be a graph class.  We define the \emph{$\CCC$\hy
    cover number} of a graph to be a the smallest number of modules the vertex
  set can be partitioned into such that each module induces a subgraph that
  belongs to the class $\CCC$.

  We show that each graph problem that can be expressed in Monadic
  Second Order (MSO) logic has a polynomial kernel with a linear
  number of vertices when parameterized by the $\CCC$\hy cover number
  for any fixed class $\CCC$ of bounded rank-width (or equivalently,
  of bounded clique-width, or bounded Boolean width).  Many graph
  problems such as \textsc{Independent Dominating Set},
  \textsc{$c$-Coloring}, and \textsc{$c$-Domatic Number} are covered
  by this meta-kernelization result.

  Our second result applies to MSO expressible optimization problems,
  such as \textsc{Minimum Vertex Cover}, \textsc{Minimum Dominating
    Set}, and \textsc{Maximum Clique}. We show that these problems
  admit a polynomial annotated kernel  with a linear
  number of vertices.
\end{abstract}

\section{Introduction}
Kernelization is an algorithmic technique that has become the subject
of a very active field in parameterized complexity, see, e.g., the
references in
\cite{Fellows06,GuoNiedermeier07,Rosamond10}. Kernelization can be
considered as a \emph{preprocessing with performance guarantee} that
reduces an instance of a parameterized problem in polynomial time to a
decision-equivalent instance, the \emph{kernel}, whose size is bounded
by a function of the parameter
alone~\cite{Fellows06,GuoNiedermeier07,Fomin10}; if the reduced
instance is an instance of a different problem, then it is called a
\emph{bikernel}.  Once a kernel or bikernel is obtained, the time
required to solve the original instance is bounded by a function of
the parameter and therefore independent of the input
size. Consequently one aims at (bi)kernels that are as small as
possible.

Every fixed-parameter tractable problem admits a kernel, but the size of the
kernel can have an exponential or even non-elementary dependence on the
parameter~\cite{FlumGrohe06}. Thus research on kernelization is typically
concerned with the question of whether a fixed-parameter tractable problem
under consideration admits a small, and in particular a \emph{polynomial},
kernel.  For instance, the parameterized \textsc{Minimum Vertex Cover} problem
(does a given graph have a vertex cover consisting of $k$ vertices?)  admits a
polynomial kernel containing at most $2k$ vertices. There are many
fixed-parameter tractable problems for which no polynomial kernels are known.
Recently, theoretical tools have been developed to provide strong theoretical
evidence that certain fixed-parameter tractable problems do not admit
polynomial kernels~\cite{BodlaenderDowneyFellowsHermelin09}. In particular,
these techniques can be applied to a wide range of graph problems
parameterized by treewidth and other width parameters such as clique-width, or
rank-width.  Thus, in order to get polynomial kernels, structural parameters
have been suggested that are somewhat weaker than treewidth, including the
vertex cover number, max-leaf number, and neighborhood
diversity~\cite{FellowsJansenRosamond13,Lampis12}.  The general aim is to find
a parameter that admits a polynomial kernel while being as general as
possible.

We extend this line of research by using results from modular
decompositions and rank-width to introduce new structural parameters for which
large classes of problems have polynomial kernels. Specifically, we study the
\emph{rank-width-$d$ cover number}, which is a special case of a
\emph{$\CCC$-cover number} (see Section~\ref{section:rwc} for definitions).
We establish the following result which is an important prerequisite for our
kernelization results.
\begin{restatable}{theorem}{thecomputecover}\label{the:compute-cover}
  For every constant $d$, a smallest rank-width-$d$ cover of a graph
  can be computed in polynomial time.
\end{restatable}
Hence, for graph problems parameterized by rank-width-$d$ cover number,
we can always compute the parameter in polynomial time.  The proof of
Theorem~\ref{the:compute-cover} relies on a combinatorial property of
modules of bounded rank-width that amounts to a variant of
partitivity~\cite{CheinHabibMaurer81}.

Our kernelization results take the shape of \emph{algorithmic
  meta-theorems}, stated in terms of the evaluation of formulas of
monadic second order logic (MSO) on graphs. Monadic second order logic
over graphs extends first order logic by variables that may range over
sets of vertices (sometimes referred to as MSO${}_1$ logic).
Specifically, for an MSO formula $\varphi$, our first meta-theorem
applies to all problems of the following shape, which we simply call
\emph{MSO model checking} problems.
\newcommand{\MSOMC}[1]{\textsc{MSO-MC${}_{#1}$}}
\newcommand{\MSOOPT}[2]{\textsc{MSO-Opt${}^{#1}_{#2}$}}
\begin{quote}
  $\MSOMC{\varphi}$\\ \nopagebreak
  \emph{Instance}: A  graph $G$.\\ \nopagebreak
  \emph{Question}: Does $G \models \varphi$ hold?
\end{quote}
Many $\NP$\hy hard graph problems can be naturally expressed as MSO
model checking problems, for instance \textsc{Independent Dominating
  Set}, \textsc{$c$-Coloring}, and \textsc{$c$-Domatic Number}.
 
\begin{restatable}{theorem}{thmmsomc}\label{the:mso-mc}
  Let $\CCC$ be a graph class of bounded rank-width.
  Every MSO model checking problem, parameterized by the $\CCC$\hy
  cover number of the input graph, has a polynomial kernel with a
  linear number of vertices.
\end{restatable}
 
While MSO model checking problems already capture many important graph
problems, there are some well-known optimization problems on graphs
that cannot be captured in this way, such as \textsc{Minimum Vertex
  Cover}, \textsc{Minimum Dominating Set}, and \textsc{Maximum
  Clique}. Many such optimization graph problems can be stated in
the following way. Let $\varphi=\varphi(X)$ be an MSO formula with one free set
variable~$X$ and $\gleq \in \{\leq,\geq\}$.
\begin{quote}
  $\MSOOPT{\gleq}{\varphi}$\\
  \nopagebreak \emph{Instance}: A graph $G$ and an integer~$r\in
  \Nat$. \\ \nopagebreak \nopagebreak \emph{Question}: Is there a set
  $S \subseteq V(G)$ such that $G \models \varphi(S)$ and $\Card{S} \gleq
  r$?
\end{quote}
We call problems of this form \emph{MSO optimization problems}.  MSO
optimization problems form a large fragment of the so-called
\emph{LinEMSO} problems~\cite{ArnborgLagergrenSeese91}. There are dozens
of well-known graph problems that can be expressed as MSO optimization problems.

We establish the following result.
\begin{restatable}{theorem}{thmmsoopt}\label{the:mso-opt}
  Let $\CCC$ be a graph class of bounded rank-width.  Every MSO
  optimization problem, parameterized by the $\CCC$\hy cover number of
  the input graph, has a polynomial bikernel with a linear number of
  vertices.
\end{restatable}
In fact, the obtained bikernel is an instance of an annotated variant of
the original MSO optimization problem~\cite{AbukhzamFernau06}. Hence,
Theorem~\ref{the:mso-opt} provides a polynomial kernel for an annotated
version of the original MSO optimization problem.

\longversion{
For obtaining the kernel for MSO model checking problems we proceed as
follows.  First we compute a smallest rank-width-$d$ cover of the
input graph $G$ in polynomial time. Second, we compute for each module
a small representative of constant size. Third, we replace each module
with a constant size module, which results in the kernel. For the MSO
optimization problems we proceed similarly. However, in order to
represent a possibly large module with a small module of constant
size, we need to keep the information how much a solution projected on
a module contributes to the full solution. We provide this information
by means of annotations to the kernel.}

We would like to point out that a class of graphs has bounded
rank-width iff it has bounded clique-width iff it has bounded
Boolean-width~\cite{BuixuanTelleVatshelle11}. Hence, we could have
equivalently stated the theorems in terms of clique-width or Boolean
width.
Furthermore we would like to point out that the theorems hold also for
some classes~$\CCC$ where we do not know whether $\CCC$ can be
recognized in polynomial time, and where we do not know how to compute
the partition in polynomial time.  For instance, the theorems hold if
$\CCC$ is a graph class of bounded clique-width (it is not known
whether graphs of clique-width at most $4$ can be recognized in
polynomial time).
\shortversion{
\paragraph{Note:} Due to space constraints, most proofs are placed in the
appendix. A full version of this paper is available on arxiv.org
(arXiv:1303.1786).
}
 
\section{Preliminaries}\label{section:prel}
The set of natural numbers (that is, positive integers) will be denoted by
$\Nat$. For $i \in \Nat$ we write $[i]$ to denote the set $\{1,
\dots, i \}$.
\paragraph{Graphs.} We will use standard graph theoretic terminology and notation
(cf. \cite{Diestel00}). A \emph{module} of a graph $G = (V,E)$ is a
nonempty set $X \subseteq V$ such that for each vertex $v \in V
\setminus X$ it holds that either no element of $X$ is a neighbor
of~$v$ or every element of $X$ is a neighbor of~$v$. We say two
modules $X, Y \subseteq V$ are \emph{adjacent} if there are vertices
$x \in X$ and $y \in Y$ such that $x$ and $y$ are adjacent. A
\emph{modular partition} of a graph $G$ is a partition $\{U_1, \dots,
U_k\}$ of its vertex set such that $U_i$ is a module of $G$ for each $i
\in [k]$.

\paragraph{Monadic Second-Order Logic on Graphs.}
We assume that we have an infinite supply of individual variables,
denoted by lowercase letters $x,y,z$, and an infinite supply of set
variables, denoted by uppercase letters $X,Y,Z$. \emph{Formulas} of
\emph{monadic second-order logic} (MSO) are constructed from atomic
formulas $E(x,y)$, $X(x)$, and $x = y$ using the connectives $\neg$
(negation), $\wedge$ (conjunction) and existential quantification
$\exists x$ over individual variables as well as existential
quantification $\exists X$ over set variables. Individual variables
range over vertices, and set variables range over sets of
vertices. The atomic formula $E(x,y)$ expresses adjacency, $x = y$
expresses equality, and $X(x)$ expresses that vertex $x$ in the set
$X$. From this, we define the semantics of monadic second-order logic
in the standard way (this logic is sometimes called $\MSO_1$).

\emph{Free and bound variables} of a formula are defined in the usual way. A
\emph{sentence} is a formula without free variables. We write $\varphi(X_1,
\dots, X_n)$ to indicate that the set of free variables of formula $\varphi$
is $\{X_1, \dots, X_n\}$. If $G = (V,E)$ is a graph and $S_1, \dots, S_n
\subseteq V$ we write $G \models \varphi(S_1, \dots, S_n)$ to denote that
$\varphi$ holds in $G$ if the variables $x_i$ are interpreted by the vertices
$v_i$ and the variables $X_j$ are interpreted by the sets $S_j$ ($i\in [n]$,
$j\in [m]$).

We review \MSO \emph{types} and \emph{games} roughly following the
presentation in \cite{Libkin04}. The \emph{quantifier rank} of an \MSO formula
$\varphi$ is defined as the nesting depth of quantifiers in $\varphi$. For
non\hy negative integers $q$ and $l$, let $\MSO_{q,l}$ consist of all \MSO
formulas of quantifier rank at most $q$ with free set variables in $\{X_1,
\dots, X_l\}$.

Let $\varphi = \varphi(X_1,\dots,X_l)$ and $\psi = \psi(X_1,\dots,X_l)$ be
\MSO formulas. We say $\varphi$ and $\psi$ are equivalent, written $\varphi
\equiv \psi$, if for all graphs $G$ and $U_1, \dots, U_l \subseteq V(G)$, $G
\models \varphi(U_1,\dots, U_l)$ if and only if $G \models \psi(U_1,\dots,
U_l)$.  Given a set $F$ of formulas, let ${F/\mathord\equiv}$ denote the set of
equivalence classes of $F$ with respect to $\equiv$. The following statement
has a straightforward proof using normal forms (see Theorem 7.5 in \cite{Libkin04} for details).
\begin{fact}\label{fact:representatives}
  Let $q$ and $l$ be non\hy negative integers. The set $\MSO_{q,l}/\mathord\equiv$ is
  finite, and given $q$ and $l$ one can effectively compute a system of
  representatives of $\MSO_{q,l}/\mathord\equiv$.
\end{fact}
We will assume that for any pair of non\hy negative integers $q$ and $l$ the
system of representatives of $\MSO_{q,l}/\mathord\equiv$ given by
Fact~\ref{fact:representatives} is fixed.
\begin{definition}[\MSO Type]
  Let $q,l$ be a non\hy negative integers. For a graph $G$ and an $l$\hy
  tuple $\vec{U}$ of sets of vertices of $G$, we define
  $\mathit{type}_q(G,\vec{U})$ as the set of formulas $\varphi \in
  \MSO_{q,l}$ such that $G \models \varphi(\vec{U})$. We call
  $\mathit{type}_q(G,\vec{U})$ the \MSO \emph{rank-$q$ type of
    $\vec{U}$ in $G$}. 
\end{definition}
It follows from Fact~\ref{fact:representatives} that up to logical
equivalence, every type contains only finitely many formulas. This
allows us to represent types using \MSO formulas as follows.

\begin{restatable}{lemma}{lemtypeformula}\label{lem:typeformula}
  Let $q$ and $l$ be non\hy negative integer constants, let $G$ be a graph,
  and let $\vec{U}$ be an $l$\hy tuple of sets of vertices of $G$. One can
  effectively compute a formula $\Phi \in \MSO_{q,l}$ such that for any graph
  $G'$ and any $l$\hy tuple $\vec{U}'$ of sets of vertices of $G'$ we have $G'
  \models \Phi(\vec{U}')$ if and only if $\mathit{type}_q(G,\vec{U}) =
  \mathit{type}_q(G',\vec{U}')$. Moreover, if $G \models \varphi(\vec{U})$ can
  be decided in polynomial time for any fixed $\varphi \in \MSO_{q,l}$ then
  $\Phi$ can be computed in time polynomial in $\Card{V(G)}$.
\end{restatable}
\newcommand{\pflemtypeformula}[0]{\begin{proof}
  Let $R$ be a system of representatives of $\MSO_{q,l}/\mathord\equiv$ given
  by Fact~\ref{fact:representatives}. Because $q$ and $l$ are constant, we can
  consider both the cardinality of $R$ and the time required to compute it as
  constants. Let $\Phi \in \MSO_{q,l}$ be the formula defined as $\Phi =
  \bigwedge_{\varphi \in S} \varphi \wedge \bigwedge_{\varphi \in R \setminus
    S} \neg \varphi$, where $S = \SB \varphi \in R \SM G \models
  \varphi(\vec{U}) \SE$. We can compute $\Phi$ by deciding $G \models
  \varphi(\vec{U})$ for each $\varphi \in R$. Since the number of formulas in
  $R$ is a constant, this can be done in polynomial time if $G \models
  \varphi(\vec{U})$ can be decided in polynomial time for any fixed $\varphi
  \in \MSO_{q,l}$.

  Let $G'$ be an arbitrary graph and $\vec{U}'$ an $l$\hy tuple of subsets of
  $V(G')$. We claim that $\mathit{type}_q(G, \vec{U}) = \mathit{type}_q(G',
  \vec{U'})$ if and only if $G' \models \Phi(\vec{U}')$. Since $\Phi \in
  \MSO_{q,l}$ the forward direction is trivial. For the converse, assume
  $\mathit{type}_q(G, \vec{U}) \neq \mathit{type}_q(G', \vec{U'})$. First
  suppose $\varphi \in \mathit{type}_q(G, \vec{U}) \setminus
  \mathit{type}_q(G', \vec{U'})$. The set $R$ is a system of representatives
  of $\MSO_{q,l}/\mathord\equiv$ , so there has to be a $\psi \in R$ such that
  $\psi \equiv \varphi$. But $G' \models \Phi(\vec{U}')$ implies $G' \models
  \psi(\vec{U}')$ by construction of $\Phi$ and thus $G' \models
  \varphi(\vec{U}')$, a contradiction. Now suppose $\varphi \in
  \mathit{type}_q(G', \vec{U}') \setminus \mathit{type}_q(G, \vec{U})$. An
  analogous argument proves that there has to be a $\psi \in R$ such that
  $\psi \equiv \varphi$ and $G' \models \neg \psi(\vec{U}')$. It follows that
  $G' \not \models \varphi(\vec{U}')$, which again yields a contradiction.
\end{proof}}
\longversion{\pflemtypeformula}

\begin{definition}[Partial isomorphism]\label{def:partialisomorphism}
  Let $G, G'$ be graphs, and let $\vec{V} = (V_1, \dots, V_l)$ and
  $\vec{U} = (U_1, \dots, U_l)$ be tuples of sets of vertices with
  $V_i \subseteq V(G)$ and $U_i \subseteq V(G')$ for each $i \in
  [l]$. Let $\vec{v} = (v_1, \dots, v_m)$ and $\vec{u} = (u_1, \dots,
  u_m)$ be tuples of vertices with $v_i \in V(G)$ and $u_i \in V(G')$
  for each $i \in [m]$. Then $(\vec{v}, \vec{u})$ defines a
  \emph{partial isomorphism between $(G, \vec{V})$ and $(G',
  \vec{U})$} if the following conditions hold:
  \begin{itemize}
    \item For every $i,j \in [m]$,
    \begin{align*}
      v_i = v_j \: \Leftrightarrow \: u_i = u_j \text{ and }
      v_iv_j \in E(G)\: \Leftrightarrow \: u_iu_j \in E(G').
    \end{align*}
    \item For every $i \in [m]$ and $j \in [l]$,
      \begin{align*}
        v_i \in V_j \: \Leftrightarrow u_i \in U_j.
      \end{align*}
    \end{itemize}
\end{definition}

\begin{definition}
  Let $G$ and $G'$ be graphs, and let $\vec{V_0}$ be a $k$\hy tuple of subsets
  of $V(G)$ and let $\vec{U_0}$ be a $k$\hy tuple of subsets of $V(G')$. Let
  $q$ be a non\hy negative integer. The \emph{$q$\hy round \MSO game on $G$
    and $G'$ starting from $(\vec{V_0}, \vec{U_0})$} is played as follows.
  The game proceeds in rounds, and each round consists of one of the following
  kinds of moves.
\begin{itemize}
  \item \textbf{Point move} The spoiler picks a vertex in either $G$ or $G'$; the duplicator responds by picking a vertex in the other graph.
  \item \textbf{Set move} The spoiler picks a subset of $V(G)$ or a
    subset of $V(G')$; the duplicator responds with a subset of the
    vertex set of the other graph.
  \end{itemize}
  Let $v_1,\dots,v_m \in V(G)$ and $u_1,\dots,u_m \in V(G')$ be the point
  moves played in the $q$-round game, and let $V_1, \dots, V_l \subseteq
  V(G')$ and $U_1,\dots, U_l \subseteq V(G)$ be the set moves played in the
  $q$\hy round game, so that $l + m = q$ and moves belonging to same round
  have the same index. Then the duplicator wins the game if $(\vec{v},
  \vec{u})$ is a partial isomorphism of $(G, \vec{V_0}\vec{V})$ and $(G',
  \vec{U_0}\vec{U})$. If duplicator has a winning strategy, we write $(G,
  \vec{V_0}) \equiv^{\MSO}_q (G', \vec{U_0})$.
\end{definition}

\begin{theorem}[\cite{Libkin04}, Theorem 7.7]\label{thm:msogames} Given two graphs $G$ and $G'$ and two $l$\hy tuples $\vec{V_0}, \vec{U_0}$ of sets of vertices of $G$ and $G'$, we have \begin{align*}
    \mathit{type}_q(G, \vec{V_0}) = \mathit{type}_q(G, \vec{U_0}) \:
    \Leftrightarrow \: (G, \vec{V_0}) \equiv^{\MSO}_q (G', \vec{U_0}).
\end{align*}

\end{theorem}

\paragraph{Fixed-Parameter Tractability and Kernels.}

\newcommand{\PP}{P}
\newcommand{\QQ}{Q}

A \emph{parameterized problem} $\PP$ is a subset of $\Sigma^* \times
\Nat$ for some finite alphabet $\Sigma$. For a problem instance $(x,k)
\in \Sigma^* \times \Nat$ we call $x$ the main part and $k$ the
parameter.  A parameterized problem $\PP$ is \emph{fixed-parameter
  tractable} (FPT) if a given instance $(x, k)$ can be solved in time
$O(f(k) \cdot p(\Card{x}))$ where $f$ is an arbitrary computable
function of $k$ and $p$ is a polynomial in the input size $\Card{x}$.

A \emph{bikernelization} for a parameterized problem $\PP \subseteq
\Sigma^* \times \Nat$ into a parameterized problem $\QQ \subseteq
\Sigma^* \times \Nat$ is an algorithm that, given $(x, k) \in \Sigma^*
\times \Nat$, outputs in time polynomial in $\Card{x}+k$ a pair $(x',
k') \in \Sigma^* \times \Nat$ such that (i)~$(x,k) \in \PP$ if and
only if $(x',k') \in \QQ$ and (ii)~$\Card{x'}+k′'\leq g(k)$, where $g$
is an arbitrary computable function.  The reduced instance $(x',k')$
is the \emph{bikernel}.  If $\PP=\QQ$, the reduction is called a
\emph{kernelization} and $(x',k')$ a \emph{kernel}.  The function $g$
is called the \emph{size} of the (bi)kernel, and if $g$ is a
polynomial then we say that $\PP$ admits a \emph{polynomial
  (bi)kernel}.
  
It is well known that every fixed-parameter tractable problem admits a generic
kernel, but the size of this kernel can have an exponential or even
non-elementary dependence on the parameter~\cite{DowneyFellowsStege99}. Since
recently there have been workable tools available for providing strong
theoretical evidence that certain parameterized problems do not admit a
polynomial kernel
\cite{BodlaenderDowneyFellowsHermelin09,MisraRamanSaurabh11}.

\paragraph{Rank-width}

The graph invariant rank-width was introduced by Oum and Seymour~\cite{OumSeymour06} with the original intent of investigating the
graph invariant clique-width. It later turned out that rank-width
itself is a useful parameter, with several advantages over
clique-width. 

\shortversion{
For a graph $G$ and $U,W\subseteq V(G)$, let $\mx A_G[U,W]$ denote the
$U\times W$-submatrix of the adjacency matrix over the two-element
field $\mathrm{GF}(2)$, i.e., the entry $a_{u,w}$, $u\in U$ and $w\in
W$, of $\mx A_G[U,W]$ is $1$ if and only if $\{u,w\}$ is an edge
of~$G$.  The {\em cut-rank} function $\rho_G$ of a graph $G$ is
defined as follows: For a bipartition $(U,W)$ of the vertex
set~$V(G)$, $\rho_G(U)=\rho_G(W)$ equals the rank of $\mx A_G[U,W]$
over $\mathrm{GF}(2)$. 

A \emph{rank-decomposition} of a graph $G$ is a pair $(T,\mu)$
where $T$ is a tree of maximum degree 3 and $\mu:V(G)\rightarrow \{t:
\text{$t$ is a leaf of $T$}\}$ is a bijective function. For an edge
$e$ of $T$, the connected components of $T\setminus e$ induce a
bipartition $(X,Y)$ of the set of leaves of $T$.  The \emph{width} of
an edge $e$ of a rank-decomposition $(T,\mu)$ is $\rho_G(\mu^{-1} (X))$.
The \emph{width} of $(T,\mu)$ is the maximum width over all edges of
$T$.  The \emph{rank-width} of $G$ is the minimum width over all
rank-decompositions of $G$.  
}

\longversion{
A set function $f:2^M\rightarrow \mathbb{Z}$ is called
\emph{symmetric} if $f(X)=f(M\setminus X)$ for all $X\subseteq M$.
For a symmetric function $f:2^M\rightarrow \mathbb{Z}$ on a finite
set~$M$, a \emph{branch-decomposition} of $f$ is a pair $(T,\mu)$
where $T$ tree of maximum degree 3 and $\mu:M\rightarrow \{t:
\text{$t$ is a leaf of $T$}\}$ is a bijective function.  For an edge
$e$ of $T$, the connected components of $T\setminus e$ induce a
bipartition $(X,Y)$ of the set of leaves of $T$.  The \emph{width} of
an edge $e$ of a branch-decomposition $(T,\mu)$ is $f(\mu^{-1} (X))$.
The \emph{width} of $(T,\mu)$ is the maximum width over all edges of
$T$.  The \emph{branch-width} of $f$ is the minimum width over all
branch-decompositions of $f$.  If $|M|\le 1$, then we define the
branch-width of $f$ as $f(\emptyset)$.  A natural application of this
definition is the branch-width of a graph, as introduced by Robertson
and Seymour~\cite{RobertsonSeymour91}, where $M=E(G)$, and $f$ the
connectivity function of~$G$.
 
There is, however, another interesting application of the
aforementioned general notions, in which we consider the vertex set 
$V(G)=M$ of a graph $G$ as the ground set.

For a graph $G$ and $U,W\subseteq V(G)$, let $\mx A_G[U,W]$ denote the
$U\times W$-submatrix of the adjacency matrix over the two-element
field $\mathrm{GF}(2)$, i.e., the entry $a_{u,w}$, $u\in U$ and $w\in
W$, of $\mx A_G[U,W]$ is $1$ if and only if $\{u,w\}$ is an edge
of~$G$.  The {\em cut-rank} function $\rho_G$ of a graph $G$ is
defined as follows: For a bipartition $(U,W)$ of the vertex
set~$V(G)$, $\rho_G(U)=\rho_G(W)$ equals the rank of $\mx A_G[U,W]$
over $\mathrm{GF}(2)$.  A \emph{rank-decomposition} and
\emph{rank-width} of a graph $G$ is the branch-decomposition and
branch-width of the cut-rank function
$\rho_G$ of~$G$ on $M=V(G)$, respectively.
}

\begin{theorem}[\cite{HlinenyOum08}]\label{thmrankdecomp} Let $k \in \Nat$ be a constant and
  $n \geq 2$. For an $n$-vertex graph $G$, we can output a
  rank-decomposition of width at most $k$ or confirm that the
  rank-width of $G$ is larger than $k$ in time $O(n^3)$.
\end{theorem}

\begin{theorem}[\cite{GanianHlineny10}]\label{thm:msorankwidth}
  Let $d \in \Nat$ be a constant and let $\varphi$ and $\psi=\psi(X)$ be fixed
  \MSO formulas. Given a graph $G$ with $\rw(G) \leq d$, one can decide
  whether $G \models \varphi$ in polynomial time. Moreover, a set $S \subseteq
  V(G)$ of minimum (maximum) cardinality such that $G \models \psi(S)$ can be
  found in polynomial time, if one exists.
\end{theorem}

\section{Rank-width Covers}\label{section:rwc}

Let $G^1$ be the trivial single-vertex graph, and let $\CCC$ be a graph class
such that $G^1\in \CCC$.  We define a \emph{$\CCC$\hy cover of $G$} as a
modular partition $\{U_1,\dots,U_k\}$ of $V(G)$ such that the induced subgraph $G[U_i]$ belongs to
the class $\CCC$ for each $i \in [k]$. Accordingly, the \emph{$\CCC$\hy cover number} of $G$ is the
size of a smallest $\CCC$\hy cover of~$G$. 

Of special interest to us are the classes $\RRR_d$ of graphs of
\emph{rank-width} at most~$d$. We call the $\RRR_d$\hy cover number also the
\emph{rank-width-$d$ cover number}.  If $\CCC$ is the class of complete and
edgeless graphs, then the $\CCC$\hy cover number equals the neighborhood
diversity~\cite{Lampis12}, and clearly $\CCC\subsetneq \RRR_1$.
 Figure~\ref{fig:diagram} shows the relationship between the
rank-width-$d$ cover number and some other graph invariants.


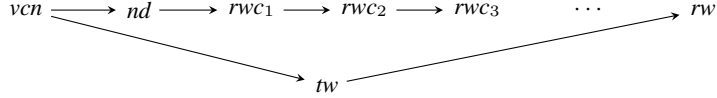
\begin{figure}[tb]
  \centering

   \begin{tikzpicture}
     \draw 

     (-1,0)      node (vcn) {$\vcn$} 
     (.5,0)      node (nd) {$\nd$} 
     (2,0)      node (r1) {$\rwc{1}$} 
     (3.5,0)      node (r2) {$\rwc{2}$} 
     (5,0)      node (r3) {$\rwc{3}$} 
     (6.5,0)      node (dots) {$\cdots$} 
     (8,0)      node (rw) {$\rw$} 

     (3,-1)      node (tw) {$\tw$} 

;
    \draw[->,shorten >=1pt,>=stealth]    
    (vcn) edge (nd)
    (nd) edge (r1)
    (r1) edge (r2)
    (r2) edge (r3)

    (vcn) edge (tw)
    (tw) edge (rw)
;
 
       \end{tikzpicture}%

       \caption{Relationship between graph invariants: the vertex
         cover number ($\vcn$), the neighborhood diversity ($\nd$), the
         rank-width-$d$ cover number ($\rwc{d}$), the rank-width
         ($\rw$), and the treewidth ($\tw$). An arrow from $A$ to $B$
         indicates that for any  graph class for which $B$ is bounded also
         $A$ is bounded. }
  \label{fig:diagram}
\end{figure}

We state some further properties of rank-width-$d$ covers.
\begin{proposition}
  Let $\vcn$, $\nd$, and $\rw$ denote the vertex cover number, the
  neighborhood diversity, and the rank-width of a graph $G$,
  respectively. Then the following (in)equalities hold for any $d \in
  \Nat$:
\begin{enumerate}
\item $\rwc{d}(G)\leq  \nd(G)\leq 2^{\vcn(G)}$, \label{nd1}
\item if $d\geq rw(G)$, then $|\rwc{d}(G)|=1$. \label{nd2}
\end{enumerate}
\end{proposition}

\begin{proof}
  (\ref{nd1})~The neighborhood diversity of a graph is also a rank-width-$1$
  cover. The neighborhood diversity is known to be upper-bounded by
  $2^{\vcn(G)}$ \cite{Lampis12}.
  (\ref{nd2})~follows immediately from the definition of
  rank-width-$d$ covers.
\end{proof}

\subsection{Finding the Cover}

Next we state several properties of modules of graphs. These will be
used to obtain a polynomial algorithm for finding smallest
rank-width-$d$ covers.

The \emph{symmetric difference} of sets $A,B$ is $A \mathop{\triangle}
B = (A \setminus B) \cup (B \setminus A)$. Sets $A,B$ \emph{overlap}
if $A \cap B \neq \emptyset$ but neither $A \subseteq B$ nor $B
\subseteq A$.

\begin{definition}\label{def:partitive} Let $\mathcal{S} \subseteq 2^S$ be a family of
  subsets of a set $S$. We call $\mathcal{S}$ \emph{partitive} if it
  satisfies the following properties:
\begin{enumerate}
\item $S \in \mathcal{S}$, $\emptyset \notin \mathcal{S}$, and $\{x\} \in \mathcal{S}$ for each $x \in S$.\
\item For every pair of overlapping subsets $A, B \in \mathcal{S}$,
  the sets $A \cup B, A \cap B, A \mathop{\triangle} B, A \setminus~B$,
  and $B \setminus A$ are contained in $\mathcal{S}$.
  \end{enumerate}
\end{definition}

\begin{theorem}[\cite{CheinHabibMaurer81}]\label{thm:modpartitive}
  The family of modules of a graph $G$ is partitive.
\end{theorem}

\begin{lemma}[\cite{BuixuanHabibLimouzyDemontgolfier09}]\label{lem:minmodule} Let $G$ be a graph and $x,y \in V(G)$. There is a
  unique minimal (with respect to set inclusion) module $M$ of $G$
  such that $x, y \in M$, and $M$ can be computed in time
  $O(\Card{V(G)}^2)$.
\end{lemma}
\begin{restatable}{lemma}{lemunion}\label{lem:union}
  Let $d \in \Nat$ be a constant. Let $G$ be a graph and let $M_1, M_2$
  be modules of $G$ such that $M_1 \cap M_2 \neq \emptyset$ and $\max(rw(G[M_1]),
  rw(G[M_2])) \leq d$. Then $M_1
 \cup M_2$ is a module of $G$ and $rw(G[M_1
  \cup M_2]) \leq d$.
\end{restatable}
\newcommand{\pflemunion}[0]{
\begin{proof}
  If $M_1 \subseteq M_2$ or $M_2 \subseteq M_1$ the result is
  immediate. Suppose $M_1$ and $M_2$ overlap and let
  $M_{11} = M_1 \setminus M_2, M_{22} = M_2 \setminus M_1$, and
 $M_{12} =
  M_1 \cap M_2$. It follows from Theorem~\ref{thm:modpartitive} that these
  sets are modules of $G$. Let $v_{11} \in M_{11}, v_{22} \in M_{22}$, and
  $v_{12} \in M_{12}$. We show that $rw(G[M_1 \cup M_2]) \leq d$. By
  assumption, both
 $G[M_1]$ and $G[M_2]$ have rank-width at most $d$. Since
  rank-width
 is preserved by taking induced subgraphs, the graphs $G_{11} =
  G[M_{11} \cup \{v_{12}\}]$, $G_{12} = G[M_{12} \cup \{v_{22}\}]$,
 and
  $G_{22} = G[M_{22} \cup \{v_{12}\}]$ also have rank-width at most
 $d$. Let
  $\mathcal{T}_{11} = (T_{11}, \mu_{11})$, $\mathcal{T}_{12}
 = (T_{12},
  \mu_{12})$, and $\mathcal{T}_{22} = (T_{22}, \mu_{22})$
 be witnessing rank
  decompositions of $G_{11}, G_{12}$, and $G_{22}$,
 respectively.
 
 We construct a rank decomposition $\mathcal{T} = (T, \mu)$ of $G[M_1 \cup
 M_2]$ as follows. Let $l_{22}$ be the leaf (note that $\mu_{12}$ is
 bijective) of $T_{12}$ such that $\mu_{12}(v_{22}) = l_{22}$. Moreover, let
 $l_{12}$ and $l_{12}'$ be the leaves of $T_{11}$ and $T_{22}$ such that
 $\mu_{11}(v_{12}) = l_{12}$ and $\mu_{22}(v_{12}) = l_{12}'$,
 respectively. We obtain $T$ from $T_{12}$ by adding disjoint copies of
 $T_{11}$ and $T_{22}$ and then identifying $l_{22}$ with the copies of
 $l_{12}$ and $l_{12}'$. Since $T_{11}, T_{12}$, and $T_{22}$ are subcubic, so
 is $T$.
 
 We define the mapping $\mu: M_1 \cup M_2 \rightarrow \SB t
  \SM $ t
 is a leaf of $T \SE$ by
  \begin{align*}
    \mu(v) = \begin{cases}
      \mu_{12}(v) &\text{if $v \in M_{12}$,} \\
      c(\mu_{11}(v)) &\text{if $v \in M_{11}$,} \\
      c(\mu_{22}(v)) &\text{otherwise,}
    \end{cases}
  \end{align*}
  where $c$ maps nodes in $T_{11} \cup T_{22}$ to their copies in
  $T$. The mappings $\mu_{11}, \mu_{12}$, and $\mu_{22}$ are
  bijections and $c$ is injective, so $\mu$ is injective. By
  construction, the image of $M_1 \cup M_2$ under $\mu$ is the set of
  leaves of $T$, so $\mu$ is a bijection. Thus $\mathcal{T} = (T,
  \mu)$ is a rank decomposition of $G[M_1 \cup M_2]$.
 
  We prove that the width of $\mathcal{T}$ is at most $d$. Given a rank
  decomposition $\mathcal{T}^* = (T^*, \mu^*)$ and an edge $e \in T^*$, the
  connected components of $T^* \setminus \{e\}$ induce a bipartition $(X, Y)$
  of the leaves of $T^*$. We set $f: (\mathcal{T}^*, e) \mapsto
  ({\mu^*}^{-1}(X), {\mu^*}^{-1}(Y))$. Take any edge $e$ of $T$. There is a
  natural bijection $\beta$ from the edges in $T$ to the edges of $T_{11} \cup
  T_{12} \cup T_{22}$. Accordingly, we distinguish three cases for $e ' =
  \beta(e)$:
 
  \begin{enumerate}
  \item $e' \in T_{11}$. Let $(U, W) = f(\mathcal{T}_{11},
    e')$. Without loss of generality assume that $v_{12} \in
    W$. Then by construction of $\mathcal{T}$, we have $f(\mathcal{T},
    e) = (U, W \cup M_2)$. Pick any $u \in U \subseteq M_{11}$ and $v
    \in M_2 \setminus W$. Since $M_2$ is a module of $G$ with $v,
    v_{12} \in M_2$ but $u \notin M_2$ we have $\mathbf{A}_G(u, v) =
    \mathbf{A}_G(u, v_{12})$. As a consequence, $\mathbf{A}_G[U, W \cup
    M_2]$ can be obtained from $\mathbf{A}_G[U, W]$ by copying the
    column corresponding to $v_{12}$. This does not increase the rank
    of the matrix. \label{rwmodule:case1}
  \item $e' \in T_{22}$. This case is symmetric to case~\ref{rwmodule:case1}, with
    $M_{22}$ and $M_1$ taking the roles of $M_{11}$ and $M_2$, respectively.
  \item $e' \in T_{12}$. Let $(U, W) = f(\mathcal{T}_{12},
    e')$. Without loss of generality assume that $v_{22} \in W$. Then
    $f(\mathcal{T}, e) = (U, W \cup M_{11} \cup M_{22})$. Let $u \in U
    \subseteq M_{12}$ and $v \in M_{22}$. Since $M_1$ is a module and $u \in
    M_1$ but $v, v_{22} \notin M_1$, we must have $\mathbf{A}_G(u, v) =
    \mathbf{A}_G(u, v_{22})$, so one can simply copy the column corresponding to
    $v_{22}$. Now consider $w \in M_{11}$. Suppose $wu \in E(G)$. Since $u,
    v_{22} \in M_2$ but $w \notin M_2$, we must have $wv_{22} \in E(G)$
    because $M_2$ is a module. Then since $w, u \in M_1$ and $v_{22} \notin
    M_1$ we must have $uv_{22} \in E(G)$ because $M_1$ is a module. A
    symmetric argument proves that $uv_{22} \in E(G)$ implies $wu \in
    E(G)$. It follows that $\mathbf{A}_G(u, w) = \mathbf{A}_G(u, v_{22})$. So
    again $\mathbf{A}_G[U, W \cup M_{11} \cup M_{22}]$ can be obtained from
    $\mathbf{A}_G[U, W]$ by copying columns, and thus the two matrices have the
    same rank.
  \end{enumerate}
  Since $\beta$ is bijective, this proves that the rank of any bipartite
  adjacency matrix induced by removing an edge $e \in T$ is bounded by $d$. We
  conclude that the width of $\mathcal{T}$ is at most $d$ and thus
  \mbox{$rw(G[M_1 \cup M_2]) \leq d$}.
\end{proof}}
\longversion{\pflemunion}

\begin{definition}
  Let $G$ be a graph and $d \in \Nat$.  We define a relation
  $\sim^G_d$ on $V(G)$ by letting $v \sim^G_d w$ if and only if there
  is a module $M$ of $G$ with $v,w \in M$ and $rw(G[M]) \leq d$. We
  drop the superscript from $\sim^G_d$ if the graph $G$ is clear from
  context.
\end{definition}

\begin{restatable}{proposition}{propeqrel}\label{prop:eqrel}
  For every graph $G$ and $d \in \Nat$ the relation $\sim_d$ is an
  equivalence relation, and each equivalence class $U$ of $\sim_d$ is a module
  of $G$ with $rw(G[U]) \leq d$.
\end{restatable}
\newcommand{\pfpropeqrel}[0]{
\begin{proof} 
  Let $G$ be a graph and $d \in \Nat$. For every $v \in V(G)$, the
  singleton $\{v\}$ is a module of $G$, so $\sim_d$ is reflexive. Symmetry of
  $\sim_d$ is trivial. For transitivity, let $u,v,w \in V(G)$ such that $u
  \sim_d v$ and $v \sim_d w$. Then there are modules $M_1, M_2$ of $G$ such
  that $u,v \in M_1$, $v, w \in M_2$, and $rw(G[M_1]), rw(G[M_2]) \leq d$. By
  Lemma~\ref{lem:union} $M_1 \cup M_2$ is a module of $G$ with $rw(G[M_1 \cup
  M_2]) \leq d$. In combination with $u,w \in M_1 \cup M_2$ that implies $u
  \sim_d w$. This concludes the proof that $\sim_d$ is an equivalence
  relation.
  
  Now let $v \in V(G)$ and let $U = [v]_{\sim_d}$. For each $u \in U$ there is
  a module $W_u$ of $G$ with $u,v \in W_u$ and $rw(G[W_u]) \leq d$. By
  Lemma~\ref{lem:union}, $W = \bigcup_{u \in U} W_u$ is a module of $G$ and
  $rw(G[W]) \leq d$.  Clearly, $[v]_{\sim_d} \subseteq W$. On the other hand,
  $u \in W$ implies $v \sim_d u$ by definition of $\sim_d$, so $W \subseteq
  [v]_{\sim_d}$. That is, $W = [v]_{\sim_d}$.
\end{proof}}
\longversion{\pfpropeqrel}
\begin{restatable}{corollary}{cormincover}\label{cor:mincover}
  Let $G$ be a graph and $d \in \Nat$. 
  The equivalence classes of $\sim_d$ form a smallest  rank-width-$d$
  cover of $G$.
\end{restatable}
\newcommand{\pfcormincover}[0]{
\begin{proof}
  Let $\mathcal{U} = \{U_1,\dots, U_k\}$ be the set of equivalence
  classes of $\sim_d$.  It is immediate from
  Proposition~\ref{prop:eqrel} that $\mathcal{U}$ is a rank-width-$d$
  cover of $G$. Let $\mathcal{V} = \{V_1, \dots, V_j\}$ be a partition
  of $V(G)$ with $j < k$. By the pigeonhole principle, there have to
  be vertices $v_1, v_2 \in V(G)$ and indices $i_1, i_2 \in [k]$, $i
  \in [j]$ such that $v_1, v_2 \in V_j$ but $v_1 \in U_{i_1}$ and $v_2
  \in U_{i_2}$, where $i_1 \neq i_2$. Thus $v_1 \nsim_d v_2$, so there
  is no module $M$ of $V(G)$ such that $v_1, v_2 \in M$ and $rw(G[M])
  \leq d$. In particular, $V_i$ is not a module or $rw(G[V_i]) >
  d$. So $\mathcal{V}$ is not a rank-width-$d$ cover of $G$.
\end{proof}}
\longversion{\pfcormincover}

\begin{restatable}{proposition}{propequdecision}\label{prop:equdecision}
  Let $d \in \Nat$ be a constant. Given a graph $G$ and two
  vertices $v, w \in V(G)$, we can decide whether $v \sim_d w$ in
  polynomial time.
\end{restatable}
\newcommand{\pfpropequdecision}[0]{
\begin{proof}
  By Lemma~\ref{lem:minmodule} we can compute the unique minimal (with
  respect to set inclusion) module $M$ containing $v$ and $w$ in time
  $O(\Card{V(G)}^2)$. Since rank-width is preserved for induced subgraphs,
  there is a module $M'$ containing $v$ and $w$ with $rw(G[M']) \leq d$ if and
  only if $rw(G[M]) \leq d$. By Theorem~\ref{thmrankdecomp} this can be
  decided in time $O(\Card{V(G)}^3)$.
\end{proof}}
\longversion{\pfpropequdecision}

\begin{proof}[of Theorem~\ref{the:compute-cover}]
  Let $d \in \Nat$ be a constant. Given a graph $G$, we can compute the set of
  equivalence classes of $\sim_d$ by testing whether $v \sim_d w$ for each
  pair of vertices $v,w \in V(G)$. By Proposition~\ref{prop:equdecision}, this
  can be done in polynomial time, and by Corollary~\ref{cor:mincover},
  $V(G)/\mathord \sim_d$ is a smallest rank-width-$d$ cover of $G$.
\end{proof}

\section{Kernels for MSO Model Checking}\label{section:msokernel}
In this section, we show that every \MSO model checking problem admits a
polynomial kernel when parameterized by the $\CCC$\hy cover number of the
input graph, where $\CCC$ is some recursively enumerable class of
graphs satisfying the following properties:
\begin{enumerate}[label=(\Roman*), align=left]
\item $\CCC$ contains the single\hy vertex graph, and a
  $\mathcal{C}$\hy cover of a graph $G$ with minimum cardinality can
  be computed in polynomial time. \label{msokernelc1}
    \item There is an algorithm $\mathbb{A}$ that decides whether $G
      \models \varphi$ in time polynomial in $\Card{V(G)}$ for any
      fixed \MSO sentence $\varphi$ and any graph $G \in
      \CCC$. \label{msokernelc2}
    \end{enumerate}
    Let $G$ be a graph and $U \subseteq V(G)$. Let $\vec{v}$ be an
    $m$\hy tuple of vertices of $G$, and let $\vec{V}$ be an $l$\hy
    tuple of sets of vertices of $G$. We write $\vec{V}|_U = (V_1 \cap
    U, \dots, V_l \cap U)$ to refer to the elementwise intersection of
    $\vec{V}$ with $U$. Similarly, we let $\vec{v}|_U =
    (v_{i_1},\dots,v_{i_t})$, $t \leq m$ denote the subsequence of
    elements from $\vec{v}$ contained in $U$. If $\{U_1,\dots,U_k\}$ is a
    modular partition of $G$ and $i \in [k]$ we will abuse notation
    and write $\vec{v}|_i = \vec{v}|_{U_i}$ and $\vec{V}|_i =
    \vec{V}_{U_i}$ if there is no ambiguity about what partition the index belongs to.
\begin{definition}[Congruent]\label{def:congruent}
  Let $q$ and $l$ be non\hy negative integers and let $G$ and $G'$ be graphs
  with modular partitions $\{M_1,\dots,M_k\}$ and $\{M_1',\dots,M_k'\}$,
  respectively. Let $\vec{V_0}$ be an $l$\hy tuple of subsets of $V(G)$ and let
  $\vec{U_0}$ be an $l$\hy tuple of subsets of $V(G')$. We say $(G, \vec{M},
  \vec{V_0})$ and $(G', \vec{M}', \vec{U_0})$ are \emph{$q$\hy congruent}
  if the following conditions are met:
  \begin{enumerate}
  \item For every $i,j \in [k]$ with $i \neq j$, $M_i$ and $M_j$
    are adjacent in $G$ if and only if $M_i'$ and $M_j'$ are adjacent
    in $G'$.\label{cond:modulesadjacent}
  \item For each $i \in [k]$, $\mathit{type}_q(G[M_i], \vec{V_0}|_i) =
    \mathit{type}_q(G'[M_i'],\vec{U_0}|_i)$ \label{cond:moduletypes}
\end{enumerate}
\end{definition}
\begin{restatable}{lemma}{lempartitiongame}\label{lem:partitiongame}
  Let $q$ and $l$ be non\hy negative integers and let $G$ and $G'$ be graphs with
  modular partitions $\{M_1,\dots,M_k\}$ and $\{M_1',\dots,M_k'\}$. Let
  $\vec{V_0}$ be an $l$\hy tuple of subsets of $V(G)$ and let $\vec{U_0}$
  be an $l$\hy tuple of subsets of $V(G')$. If $(G, \vec{M},
  \vec{V_0})$ and $(G',\vec{M}', \vec{U_0})$ are $q$\hy congruent,
  then $\mathit{type}_q(G, \vec{V_0}) = \mathit{type}_q(G',
  \vec{U_0})$.
  \end{restatable}
\newcommand{\pflempartitiongame}[0]{
\begin{proof}
  For $i \in [k]$, we write $G_i = G[M_i]$ and $G_i' = G'[M_i']$. By
  Theorem~\ref{thm:msogames}, Condition~\ref{cond:moduletypes} of
  Definition~\ref{def:congruent} is equivalent to $(G_i, \vec{V_0}|_i)
  \equiv^{\MSO}_q (G'_i, \vec{U_0}|_i)$. That is, for each $i \in [k]$,
  duplicator has a winning strategy $\pi_i$ in the $q$\hy round \MSO game
  played on $G_i$ and $G_i'$ starting from $(\vec{V_0}|_i, \vec{U_0}|_i)$. We
  construct a strategy witnessing $(G, \vec{V_0}) \equiv^{\MSO}_q (G',
  \vec{U_0})$ by aggregating duplicator's moves from these $k$ games in the
  following way:
    \begin{enumerate}
    \item Suppose spoiler makes a set move $W$ and assume without loss of
      generality that $W \subseteq V(G)$. For $i \in [k]$, let $W_i = M_i \cap
      W$, and let $W_i'$ be duplicator's response to $W_i$ according to
      $\pi_i$. Then duplicator responds with $W' = \cup_{i=1}^k W_i'$.
    \item Suppose spoiler makes a point move $s$ and again assume without loss
      of generality that $s \in V(G)$. Then $s \in M_i$ for some $i \in
      [k]$. Duplicator responds with $s' \in M_i'$ according to $\pi_i$.
    \end{enumerate}
    Assume duplicator plays according to this strategy and consider a play of
    the $q$\hy round \MSO game on $G$ and $G'$ starting from
    $(\vec{V_0}, \vec{U_0})$. Let $v_1,\dots,v_m\in V(G)$ and
    $u_1,\dots,u_m \in V(G')$ be the point moves and $V_1, \dots, V_l
    \subseteq V(G')$ and $U_1,\dots, U_l \subseteq V(G)$ be the set
    moves, so that $l + m = q$ and the moves made in the same round
    have the same index. We claim that $(\vec{v},\vec{u})$
    defines a partial isomorphism between $(G,\vec{V_0}\vec{V})$ and
    $(G', \vec{U_0}\vec{U})$.
    \begin{itemize}
    \item Let $j_1,j_2 \in [m]$ and let $i_1, i_2 \in [k]$ such that $v_{j_1}
      \in M_{i_1}$ and $v_{j_2} \in M_{i_2}$. Suppose $i_1 = i_2 = i$. Since
      duplicator plays according to a winning strategy in the game on $G_i$
      and $G_i'$, the restriction $(\vec{v}|_i, \vec{u}|_i)$ defines a partial
      isomorphism between $(G_i, (\vec{V_0}\vec{V})|_i)$ and $(G_i',
      (\vec{U_0}\vec{U})|_i)$. It follows that $(v_{j_1},v_{j_2}) \in E(G)$ if
      and only if $(u_{j_1},u_{j_2}) \in E(G')$ and $v_{j_1} = v_{j_2}$ if and
      only if $u_{j_1} = u_{j_2}$.  Now suppose $i_1 \neq i_2$. Then $v_{j_1}
      \neq v_{j_2}$ and also $u_{j_1} \neq u_{j_2}$ since $u_{j_1} \in
      M_{i_1}'$ and $u_{j_2} \in M_{i_2}'$ by choice of duplicator's
      strategy. By congruence, $M_{i_1}$ and $M_{i_2}$ are adjacent in $G$ if
      and only if $M_{i_1}'$ and $M_{i_2}'$ are adjacent in $G'$, so we must
      have $(v_{j_1}, v_{j_2}) \in E(G)$ if and only if $(u_{j_1}, u_{j_2})
      \in E(G')$.
    \item Let $j \in [m]$ and let $i \in [k]$ such that $v_j \in M_i$. By
      construction of duplicator's strategy, we have $u_j \in M_i'$. Note that
      if $x \in S$ then $x \in S'$ if and only if $x \in S'|_S$ for arbitrary
      sets $S$ and $S'$. Combined with the fact that $(\vec{v}|_i,
      \vec{u}|_i)$ defines a partial isomorphism between $(G_i,
      (\vec{V_0}\vec{V})|_i)$ and $(G_i', (\vec{U_0}\vec{U})|_i)$, this
      observation implies that $v_i$ is contained in any of the sets from
      $\vec{V_0}\vec{V}$ if and only if $u_i$ is contained in the sets from
      $\vec{U_0}\vec{U}$ with the same indices.
  \end{itemize}
\end{proof}}
\longversion{\pflempartitiongame}

\begin{restatable}{lemma}{lemconstantrep}\label{lem:constantrep} Let $\CCC$ be
  a recursively enumerable graph class and let $q$ be a non\hy negative integer
  constant. Let $G \in \CCC$ be a graph. If $G \models \varphi$ can be decided
  in time polynomial in $\Card{V(G)}$ for any fixed $\varphi \in \MSO_{q,0}$
  then one can in polynomial time compute a graph $G' \in \CCC$ such that
  $\Card{V(G')}$ is bounded by a constant and $\mathit{type}_q(G) =
  \mathit{type}_q(G')$.
\end{restatable}
\newcommand{\pflemconstantrep}[0]{
\begin{proof}
  By Lemma~\ref{lem:typeformula} we can compute a formula $\Phi$ capturing the
  type $T$ of $G$ in polynomial time. Given $\Phi$, a graph $G' \in \CCC$
  satisfying $\Phi$ can be effectively computed as follows. We start
  enumerating $\CCC$ and check for each graph $G' \in \CCC$ whether $G'
  \models \Phi$. If this is the case, we stop and output $G'$. Since $G
  \models \Phi$ this procedure must terminate eventually. Fixing $\CCC$ and
  the order in which graphs are enumerated, the number of graphs we have to
  check depends only on $T$. By Fact~\ref{fact:representatives} the number of
  rank $q$\hy types is finite for each $q$, so we can think of the total
  number of checks as bounded by a constant. Moreover the time spent on each
  check depends only on $T$ and the size of the graph $G'$. Because the number
  of graphs enumerated is bounded by a constant, we can think of the latter as
  bounded by a constant as well. Thus the algorithm computing a model of
  $\Phi$ runs in constant time.
\end{proof}}
\longversion{\pflemconstantrep}

\begin{restatable}{lemma}{lemcongruent}\label{lem:congruent}
  Let $q$ be a non\hy negative integer constant, and let $\CCC$ be a
  recursively enumerable
  graph class satisfying \ref{msokernelc2}. Then given a graph $G$ and a
  $\CCC$\hy cover $\{U_1,\dots,U_k\}$, one can in polynomial time compute a graph
  $G'$ with modular partition $\{U_1',\dots,U_k'\}$ such that $(G, \vec{U})$ and
  $(G', \vec{U}')$ are $q$\hy congruent and for each $i \in [k]$, $G'[U_i']
  \in \CCC$ and the number of vertices in $U_i'$ is bounded by a constant.
\end{restatable}
\newcommand{\pflemcongruent}[0]{
\begin{proof}
  For each $i \in [k]$, we compute a graph $G_i'\in \CCC$ of constant size
  with the same \MSO rank-$q$ type as $G_i = G[U_i]$. By
  Lemma~\ref{lem:constantrep}, this can be done in polynomial
  time. Now let $G'$ be the graph obtained from the disjoint union of
  the graphs $G_i'$ for $i \in [k]$ as follows. For $i \in [k]$, let
  $U_i'$ denote the set of vertices from the copy of $G_i'$. If $U_i$
  and $U_j$ are adjacent in $G$ for $i,j \in [k]$ and $i\neq j$, we
  insert an edge $vw$ for every $v \in U_i'$ and $w \in U_i'$. Then
  $U_1',\dots,U_k'$ is a modular partition of $G'$, and for $i,j
  \in[k]$ and $i \neq j$, modules $U_i$ and $U_j$ are adjacent in $G$
  if and only if $U_i'$ and $U_j'$ are adjacent in $G'$. It is readily
  verified that $(G, \vec{U})$ and $(G', \vec{U}')$ are $q$\hy congruent.
\end{proof}}
\longversion{\pflemcongruent}
\begin{restatable}{proposition}{propmsokernel}\label{prop:msokernel} Let
  $\varphi$ be a fixed \MSO sentence. Let $\CCC$ be a recursively enumerable
  graph class satisfiying \ref{msokernelc1} and \ref{msokernelc2}.  Then
  $\MSOMC{\varphi}$ has a polynomial kernel parameterized by the
  $\mathcal{C}$\hy cover number of the input graph.
\end{restatable}
\newcommand{\pfpropmsokernel}[0]{
\begin{proof} Let $G$ be a graph with $\mathcal{C}$\hy cover number $k$, and
  let $\{U_1, \dots, U_k\}$ be a smallest $\mathcal{C}$\hy cover given by
  \ref{msokernelc1}. Let $q$ be the quantifier rank of $\varphi$. By
  Lemma~\ref{lem:congruent} and \ref{msokernelc2}, we can in polynomial time
  compute a graph $G'$ and a modular partition $\{U_1',\dots,U_k'\}$ of $G'$ such
  that $(G, \vec{U})$ and $(G', \vec{U}')$ are $q$\hy congruent and for each
  $i \in [k]$, $\Card{U_i'}$ is bounded by a constant. It follows from
  Lemma~\ref{lem:partitiongame} that $\mathit{type}_q(G) =
  \mathit{type}_q(G')$. In particular, $G \models \varphi$ if and only if $G'
  \models \varphi$. Moreover, we have $\Card{V(G')} \in O(k)$, so $G'$ is a
  polynomial kernel.
\end{proof}}
\longversion{\pfpropmsokernel}

\begin{proof}[of Theorem~\ref{the:mso-mc}] 
  Immediate from Theorems~\ref{the:compute-cover}, \ref{thmrankdecomp}, and \ref{thm:msorankwidth}
  in combination with Proposition~\ref{prop:msokernel}.
\end{proof}

\begin{corollary} The following problems have polynomial kernels when
  parameterized by the rank\hy width\hy $d$ cover number of the input graph:
  \textsc{Independent Dominating Set}, \textsc{$c$-Coloring},
  \textsc{$c$-Domatic Number}, \textsc{$c$-Partition into Trees},
  \textsc{$c$-Clique Cover}, \textsc{$c$-Partition into Perfect
    Matchings},  \textsc{$c$-Covering by Complete Bipartite Subgraphs}.
\end{corollary}
\section{Kernels for MSO Optimization}\label{section:msooptkernel}
By definition, \MSO formulas can only directly capture decision
problems such as $3$-colorability, but many problems of interest are
formulated as optimization problems. The usual way of transforming
decision problems into optimization problems does not work here, since
the \MSO language cannot handle arbitrary numbers.

Nevertheless, there is a known solution.  Arnborg, Lagergren, and
Seese~\cite{ArnborgLagergrenSeese91} (while studying graphs of bounded
tree-width), and later Courcelle, Makowsky, and
Rotics~\cite{CourcelleMakowskyRotics00} (for graphs of bounded
clique-width), specifically extended the expressive power of MSO logic
to define so-called $\lMS$ optimization problems, and consequently
showed the existence of efficient (parameterized) algorithms for such
problems in the respective cases.

The \MSO optimization problems (problems of the form
$\MSOOPT{\gleq}{\varphi}$) considered here are a streamlined and simplified
version of the formalism introduced in
\cite{CourcelleMakowskyRotics00}. Specifically, we consider only a single free
variable $X$, and ask for a satisfying assignment of $X$ with minimum or
maximum cardinality. To achieve our results, we need a recursively enumerable
graph class $\CCC$ that satisfies \ref{msokernelc1} and \ref{msokernelc2}
along with the following property:
  \begin{enumerate}[label=(\Roman*), start=3,
    align=parleft,leftmargin = 0.8cm]
  \item Let $\varphi = \varphi(X)$ be a fixed \MSO formula. Given a graph $G
    \in \CCC$, a set $S \subseteq V(G)$ of minimum (maximum) cardinality such
    that $G \models \varphi(S)$ can be found in polynomial time, if one exists.
    \label{msooptkernelc3}
\end{enumerate}

Our approach will be similar to the \MSO kernelization algorithm, with one key
difference: when replacing the subgraph induced by a module, the cardinalities
of subsets of a given $q$\hy type may change, so we need to keep track of
their cardinalities in the original subgraph.
\newcommand{\aMSOOPT}[2]{a\textsc{MSO-Opt${}^{#1}_{#2}$}}

To do this, we introduce an annotated version of
$\MSOOPT{\gleq}{\varphi}$. Given a graph $G=(V,E)$, an \emph{annotation}
$\WWW$ is a set of triples $(X,Y,w)$ with $X\subseteq V, Y\subseteq V,
w\in \Nat$. For every set $Z\subseteq V$ we define
\[
{\WWW}(Z)=\sum_{(X,Y,w)
  \in {\WWW}, X\subseteq Z, Y\cap Z =\emptyset} w.
\]
We call the pair $(G,\WWW)$ an \emph{annotated graph}.  If the integer
$w$ is represented in binary, we can represent a triple $(X,Y,w)$ in
space $\Card{X}+\Card{Y}+\log_2(w)$.  Consequently, we may assume that
the size of the encoding of an annotated graph $(G,\WWW)$
is polynomial in $\Card{V(G)}+\Card{\WWW}+\max_{(X,Y,w)\in
  \WWW} \log_2 w$.

Each \MSO formula $\varphi(X)$ and $\gleq\in\{\leq, \geq\}$ gives rise to
an \emph{annotated MSO-optimization problem}.
\begin{quote}
  $\aMSOOPT{\gleq}{\varphi}$\\
  \nopagebreak \emph{Instance}: A graph $G$ with an annotation $\WWW$ and an integer~$r\in
  \Nat$. \\ \nopagebreak \nopagebreak \emph{Question}: Is there a set
  $Z\subseteq V(G)$ such that $G \models \varphi(Z)$ and ${\WWW}(Z)\gleq r$?
\end{quote}

Notice that any instance of $\MSOOPT{\gleq}{\varphi}$ is also an instance
of $\aMSOOPT{\gleq}{\varphi}$ with the trivial annotation
${\WWW}=\SB(\{v\},\emptyset,1) \SM v\in V(G)\SE$. The main result of
this section is a bikernelization algorithm which transforms any
instance of $\MSOOPT{\gleq}{\varphi}$ into an instance of
$\aMSOOPT{\gleq}{\varphi}$; this kind of bikernel is called an
\emph{annotated kernel} \cite{AbukhzamFernau06}. 

The results below are stated and proved for minimization problems
$\aMSOOPT{\leq}{\varphi}$ only. This is without loss of generality
-- the proofs for maximization problems are symmetric.

\begin{restatable}{lemma}{lemsettype}\label{lem:settype}
  Let $q$ and $l$ be non\hy negative integers and let $G$ and $G'$ be a
  graphs such that $G$ and $G'$ have the same $q+l$ \MSO type. Then
  for any $l$\hy tuple $\vec{V}$ of sets of vertices of $G$, there
  exists an $l$\hy tuple $\vec{U}$ of sets of vertices of $G'$ such
  that $\mathit{type}_q(G,\vec{V}) = \mathit{type}_q(G',\vec{U})$.
\end{restatable}
\newcommand{\pflemsettype}[0]{\begin{proof}
  Suppose there exists an $l$\hy tuple $\vec{V}$ of sets of vertices
  of $G$, and a formula $\varphi = \varphi(X_1,\dots,X_l) \in
  \MSO_{q,l}$ such that $G \models \varphi(V_1,\dots,V_l)$ but for
  every $l$\hy tuple $\vec{U}$ of sets of vertices of $G'$ we have $G'
  \not \models \varphi(U_1,\dots,U_l)$. Let $\psi = \exists X_1 \dots
  \exists X_l \:\varphi$. Clearly, $\psi \in \MSO_{q+l,0}$ and $G
  \models \psi$ but $G' \not \models \psi$, a contradiction.
\end{proof}}
\longversion{\pflemsettype}

\begin{restatable}{lemma}{lemsmallannotated}\label{lem:smallannotated}
  Let $\varphi = \varphi(X)$ be a fixed \MSO formula and $\CCC$ be a recursively
  enumerable graph class satisfiying \ref{msokernelc2} and
  \ref{msooptkernelc3}.  Then given an instance $(G,r)$ of
  $\MSOOPT{\leq}{\varphi}$ and a $\CCC$\hy cover $\{U_1,\dots,U_k\}$ of $G$, an
  annotated graph $(G', \WWW)$ satisfying the following properties can be
  computed in polynomial time.
  \begin{enumerate}
  \item $(G, r) \in\MSOOPT{\leq}{\varphi}$ if and only if $(G', \WWW, r) \in
  \aMSOOPT{\leq}{\varphi}$. \label{lem27cond1}
\item $\Card{V(G')} \in O(k)$. \label{lem27cond2}
\item The encoding size of $(G', \WWW)$ is $O(k\log(\Card{V(G)}))$. \label{lem27cond3}
\end{enumerate}
\end{restatable}
\newcommand{\pflemsmallannotated}[0]{
\begin{proof}
  Let $q$ be the quantifier rank of $\varphi$. By Lemma~\ref{lem:congruent}
  and \ref{msokernelc2}, we can in polynomial time compute a graph $G'$ and a
  modular partition $\{U_1',\dots,U_k'\}$ of $G'$ such that $(G, \vec{U})$ and
  $(G', \vec{U}')$ are $(q+1)$\hy congruent, $\Card{U_i'}$ is bounded by a
  constant, and $G'[U_i'] \in \CCC$ for each $i \in [k]$. To compute the
  annotation $\WWW$, we proceed as follows.  For each $i \in [k]$, we go
  through all subsets $W' \subseteq U_i'$. By Lemma~\ref{lem:typeformula}, we
  can compute a formula $\Phi$ such that for any graph $H$ and $W \subseteq
  V(H)$ we have $\mathit{type}_q(G'[U_i'], W) = \mathit{type}_q(H, W)$ if and
  only if $H \models \Phi(W)$. Since $\Card{U_i'}$ has constant size for every
  $i \in [k]$, this can be done within a constant time bound. By
  Lemma~\ref{lem:settype} and because $(G, \vec{U})$ and $(G', \vec{U}')$ are
  $(q+1)$\hy congruent, there has to be a $W \subseteq U_i$ such that $G_i
  \models \Phi(W)$. Using the algorithm given by \ref{msooptkernelc3}, we can
  compute a minimum\hy cardinality subset $W^* \subseteq U_i$ with this
  property in polynomial time. We then add the triple $(W', U_i' \setminus W',
  \Card{W^*})$ to $\WWW$. In total, the number of subsets processed is in
  $O(k)$. From this observation we get the desired bounds on the total
  runtime, $\Card{V(G')}$, and the encoding size of $(G', \WWW)$.

  We claim that $(G', \WWW, r) \in \aMSOOPT{\leq}{\varphi}$ if and only if
  $(G, r) \in \MSOOPT{\leq}{\varphi}$. Suppose there is a set $W \subseteq
  V(G)$ of vertices such that $G \models \varphi(W)$ and $\Card{W} \leq
  r$. Since $U_1,\dots,U_k$ is a partition of $V(G)$, we have $W = \cup_{i \in
    [k]} W_i$, where $W_i = W \cap U_i$. For each $i \in [k]$, let $W_i^*
  \subseteq U_i$ be a subset of minimum cardinality such that
  $\mathit{type}_q(G[U_i], W_i) = \mathit{type}_q(G[U_i], W_i^*)$. By
  Lemma~\ref{lem:settype} and $(q+1)$\hy congruence of $(G, \vec{U})$ and
  $(G', \vec{U}')$, there is $W_i' \subseteq U_i'$ for each $i \in [k]$ such
  that $\mathit{type}_q(G'[U_i'], W_i') = \mathit{type}_q(G[U_i], W_i^*)$. By
  construction, $\WWW$ contains a triple $(W_i', U_i' \setminus W_i',
  \Card{W_i^*})$. Observe that $(X, Y, w) \in \WWW$ and $(X, Y, w') \in \WWW$
  implies $w = w'$. Let $W' = \cup_{i \in [k]} W_i'$. Then by $(q+1)$\hy
  congruence of $(G, \vec{U})$ and $(G', \vec{U}')$ and
  Lemma~\ref{lem:partitiongame}, we must have $\mathit{type}_q(G, W) =
  \mathit{type}_q(G', W')$. In particular, $G' \models
  \varphi(W')$. Furthermore,
   \begin{align*}
     \WWW(W') = \sum_{(W_i', U_i' \setminus W_i', \Card{W_i^*}) \in \WWW, U_i'
       \cap W' = W'_i} \Card{W_i^*} \leq \sum_{i \in [k]} \Card{W_i} = \Card{W}
     \leq r.
   \end{align*}
   
   For the converse, let $W' \subseteq V(G')$ such that $\WWW(W') \leq r$ and
   $G' \models \varphi(W')$, let $W_i'$ denote $W' \cap U_i'$ for $i \in [k]$.
   By construction, there is a set $W_i \subseteq
   U_i$ for each $i \in [k]$ such that $\mathit{type}_q(G[U_i], W_i) =
   \mathit{type}_q(G'[U_i'], W'_i)$ and $\WWW(W') = \sum_{i \in [k]}
   \Card{W_i}$. Let $W = \cup_{i \in [k]} W_i$. Then by congruence and
   Lemma~\ref{lem:partitiongame} we get $\mathit{type}_q(G, W) =
   \mathit{type}_q(G', W')$ and thus $G \models \varphi(W)$. Moreover,
   $\Card{W} = \WWW(W') \leq r$.
 \end{proof}}
\longversion{\pflemsmallannotated}

\begin{fact}[Folklore]\label{fact:msomccomplexity}
  Given an \MSO sentence $\varphi$ and a graph $G$, one can decide
  whether $G \models \varphi$ in time $O(2^{nl})$, where $n =
  \Card{V(G)}$ and $l = \Card{\varphi}$.
\end{fact}

\begin{restatable}{proposition}{propmsooptkernel}\label{prop:msooptkernel}
  Let $\varphi = \varphi(X)$ be a fixed \MSO formula, and let $\CCC$ be a
  recursively enumerable graph class satisfying \ref{msokernelc1}, \ref{msokernelc2},
  and \ref{msooptkernelc3}. Then $\MSOOPT{\leq}{\varphi}$ has a polynomial
  bikernel parameterized by the $\CCC$\hy cover number of the input graph.
\end{restatable}
\newcommand{\pfpropmsooptkernel}[0]{
\begin{proof}
  Let $(G,r)$ be an instance of $\MSOOPT{\leq}{\varphi}$. By \ref{msokernelc1}
  a smallest $\CCC$\hy cover $\{U_1,\dots, U_k\}$ of $G$ can be computed in
  polynomial time. Let $(G',\WWW)$ be an annotated graph computed from $G$ and
  $\{U_1,\dots, U_k\}$ according to Lemma~\ref{lem:smallannotated}. Let $n =
  \Card{V(G)}$ and suppose $2^k \leq n$. Then we can solve $(G', \WWW, r)$ in
  time $n^c$ for some constant $c$ that only depends on $\varphi$ and
  $\CCC$. To do this, we go through all $2^{O(k)}$ subsets $W$ of $G'$ and
  test whether $\WWW(W) \leq r$. If that is the case, we check whether $G'
  \models \varphi(W)$. By Fact~\ref{fact:msomccomplexity} this check can be
  carried out in time $c_1 2^{c_2 k} \leq c_1 n^{c_2}$ for suitable constants
  $c_1$ and $c_2$ depending only on $\CCC$ and $\varphi$. Thus we can find a
  $c$ such that the entire procedure runs in time $n^c$ whenever $n$ is large
  enough. If we find a solution $W \subseteq V(G')$ we return a trivial yes\hy
  instance; otherwise, a trivial no\hy instance (of
  $\aMSOOPT{\leq}{\varphi}$). Now suppose $n < 2^k$. Then $\log(n) < k$ and so
  the encoding size of $\WWW$ is polynomial in $k$. Thus $(G', \WWW, r)$ is a
  polynomial bikernel.
\end{proof}}
\longversion{\pfpropmsooptkernel}

\begin{proof}[of Theorem~\ref{the:mso-opt}]
  Immediate from Theorems~\ref{the:compute-cover}, \ref{thmrankdecomp}, and
  \ref{thm:msorankwidth} when combined with
  Proposition~\ref{prop:msooptkernel}.
\end{proof}

\begin{corollary} The following problems have polynomial bikernels when
  parameterized by the rank\hy width\hy $d$ cover number of the input graph:
  \textsc{Minimum Dominating Set}, \textsc{Minimum Vertex Cover},
  \textsc{Minimum Feedback Vertex Set}, \textsc{Maximum Independent Set},
  \textsc{Maximum Clique}, \textsc{Longest Induced Path}, \textsc{Maximum
    Bipartite Subgraph}, \textsc{Minimum Connected Dominating Set}.
\end{corollary}

\section{Conclusion}
Recently Bodlaender et al.~\cite{BodlaenderEtal09} and Fomin et
al.~\cite{FominLokshtanovMisraSaurabh12} established \emph{meta-kernelization
  theorems} that provide polynomial kernels for large classes of parameterized
problems. The known meta-kernelization theorems apply to optimization problems
parameterized by \emph{solution size}. Our results are, along with very recent
results parameterized by the modulator to constant\hy treedepth
\cite{GajarskyHlinenyObdrzalek13}, the first meta-kernelization theorems that
use a \emph{structural parameter} of the input and not the solution size. In
particular, we would like to emphasize that our Theorem~\ref{the:mso-opt}
applies to a large class of optimization problems where the solution size can
be arbitrarily large.

It is also worth noting that our structural parameter, the rank-width-$d$
cover number, provides a trade-off between the maximum rank-width of modules
(the constant $d$) and the maximum number of modules (the parameter
$k$). Different problem inputs might be better suited for smaller~$d$ and
larger $k$, others for larger $d$ and smaller~$k$. This two-dimensional
setting could be seen as a contribution to a \emph{multivariate complexity
  analysis} as advocated by Fellows et al.~\cite{FellowsJansenRosamond13}.

We conclude by mentioning possible directions for future research. We believe
that some of our results can be extended from modular partitions to partitions
into splits~\cite{CharbitMontgolfierRaffinot12}.\footnote{We thank Sang-il Oum
  for pointing this out to us.} This would indeed result in a more general
parameter, however the precise details would still require further work (one
problem is that while all modules are partitive, only strong splits have this
property). Another direction would then be to focus on polynomial kernels for
problems which cannot be described by \MSO logic, such as \textsc{Hamiltonian
  Path} or \textsc{Chromatic Number}.  
\bibliographystyle{abbrv}
\bibliography{literature}

\shortversion{
\newpage

\appendix
\chapter*{Appendix}

\section*{Proof of Lemma~\ref{lem:typeformula}}
\pflemtypeformula


\section*{Proof of Lemma~\ref{lem:union}}
\pflemunion

\section*{Proof of Proposition~\ref{prop:eqrel}}
\pfpropeqrel

\section*{Proof of Corollary~\ref{cor:mincover}}
\pfcormincover

\section*{Proof of Proposition~\ref{prop:equdecision}}
\pfpropequdecision


\section*{Proof of Lemma~\ref{lem:partitiongame}}
\pflempartitiongame

\section*{Proof of Lemma~\ref{lem:constantrep}}
\pflemconstantrep

\section*{Proof of Lemma~\ref{lem:congruent}}
\pflemcongruent

\section*{Proof of Proposition~\ref{prop:msokernel}}
\pfpropmsokernel


\section*{Proof of Lemma~\ref{lem:settype}}
\pflemsettype

\section*{Proof of Lemma~\ref{lem:smallannotated}}
\pflemsmallannotated

\section*{Proof of Proposition~\ref{prop:msooptkernel}}
\pfpropmsooptkernel

}

\end{document}